\documentclass[journal]{IEEEtran}
\usepackage{latexsym}
\usepackage{graphicx}
\usepackage{array}
\usepackage{amsmath}
\usepackage{amsfonts}
\usepackage{amssymb}
\usepackage{amsthm}

\newtheorem{thm}{Theorem}

\newcommand{\bx} {\boldsymbol{x}}
\newcommand{\bA} {\boldsymbol{A}}
\newcommand{\bB} {\boldsymbol{B}}

\newcommand{\by} {\boldsymbol{y}}

\newcommand{\bW} {\boldsymbol{W}}
\newcommand{\bw} {\boldsymbol{w}}

\newcommand{\bh} {\boldsymbol{h}}
\newcommand{\bH} {\boldsymbol{H}}

\newcommand{\bxi} {\boldsymbol{\xi}}

\newcommand{\bI} {\boldsymbol{I}}
\newcommand{\bR} {\boldsymbol{R}}
\newcommand{\bU} {\boldsymbol{U}}
\newcommand{\bu} {\boldsymbol{u}}
\newcommand{\bLam} {\boldsymbol{\Lambda}}
\newcommand{\gl}{\lambda}

\def\ba#1\ea{\begin{align}#1\end{align}}

\newcommand{\bp} {\begin{proof}}
\newcommand{\ep} {\end{proof}}

\newcommand{{\bRF}} {\right\}}

\begin{document}

\title{On the Capacity of Gaussian MIMO Channels with Memory}

\author{Sergey Loyka,  Charalambos D. Charalambous

\vspace*{-1.\baselineskip}



\thanks{S. Loyka is with the School of Electrical Engineering and Computer Science, University of Ottawa, Ontario, Canada, e-mail: sergey.loyka@ieee.org}

\thanks{C.D. Charalambous is with the ECE Department, University of Cyprus, Nicosia, Cyprus, e-mail: chadcha@ucy.ac.cy}}

\maketitle


\begin{abstract}
The operational capacity of Gaussian MIMO channels with memory was obtained by Brandenburg and Wyner in [9] under certain mild assumptions on the channel impulse response and its noise covariance matrix, which essentuially require channel memory to be not too strong. This channel was also considered by Tsybakov in [10] and its information capacity was obtained in some cases. It was further conjectured, based on numerical evidence, that these capacities are the same in all cases. This conjecture is proved here. An explicit closed-form expression for the optimal input power spectral density matrix is also given. The obtained result is further extended to the case of joint constraints, including per-antenna and interference power constraints as well as energy harvesting constraints. These results imply the information-theoretic optimality of OFDM-type transmission systems for such channels with memory.
\end{abstract}

\vspace*{-0.3\baselineskip}
\begin{IEEEkeywords}
Channel capacity, MIMO, memory, OFDM.
\end{IEEEkeywords}

\vspace*{-0.5\baselineskip}
\section{Introduction}

Multi-antenna (MIMO) systems have attracted unprecedented attention in both academia and industry over the last two decades due to their large spectral efficiency and other capabilities \cite{Tse-05}-\cite{Marzetta-16}. They are now extensively used in modern cellular and WiFi networks. The capacity of AWGN MIMO channels was established in \cite{Tsybakov-65}\cite{Telatar-95} and was further extended in several directions, see e.g., \cite{Vu-11}-\cite{Loyka-20} and references therein.

While AWGN channels are memoryless, many channels in modern systems do have memory, e.g., wideband or multi-user channels, either due to the channel impulse response (in e.g., multipath channels with delay spread) or due to noise with memory (where noise also represents multi-user interference). In those cases, the memoryless results in \cite{Tsybakov-65}-\cite{Loyka-20}   do not apply. The Gaussian MIMO channel with memory was considered by Brandenburg and Wyner in \cite{Brandenburg-74} under certain (mild) assumptions on the channel impulse response and its noise covariance matrix, which essentially require the channel memory to be not too strong, see \eqref{eq.ch.c1} and \eqref{eq.ch.c2} below; these assumptions are satisfied by modern 5G channel models as in e.g., \cite{3GPP-22}. Its operational capacity\footnote{Defined as the maximum achievable transmission rate subject to the reliability criterion, i.e., arbitrary low error probability with increasing blocklength; this guarantees the existence of codebooks with as low error probability as desired \cite{Cover-06}.} under the total (average) power constraint (TPC) was established in a closed-form in \cite{Brandenburg-74} by proving direct and converse coding theorems. This channel has been also studied by Tsybakov in \cite{Tsybakov-06} and its information capacity\footnote{Defined as the maximum mutual information rate subject to the input power constraint \cite{Cover-06}; recall that operational and information capacities are not necessarily the same.} under the TPC was established in a closed-form for some special cases. Based on numerical evidence, it was further conjectured that these two capacities are always the same for the considered channel. In this Letter, we prove that this is indeed the case.

It should be emphasized that the Tsybakov's conjecture is not trivial: while the operational and information capacities are the same for information-stable channels\footnote{Recall that a channel is information stable if its information density converges to its mutual information under optimal input distribution; this is somewhat similar to the notion of ergodicity, whereby time average converges to statistical average.} \cite{Dobrushin-59}, they can be significantly different (the information capacity being larger than the operational capacity) for information-unstable channels, see e.g., \cite{Verdu}\cite{Loyka-16}\cite{Ting}. Since the channel above has memory, it is far from clear whether it is information stable or not: recall that memory can induce non-ergodic channel behaviour and many channels with memory are not information-stable; the simplest example is a non-ergodic fading channel \cite{Biglieri}.

As a by-product of our proof, a closed-form expression for the optimal (capacity-achieving) input power spectral density (PSD) matrix is obtained in the general case for the considered channel under the TPC.

This result is further extended to the case of joint power constraints, which include per-antenna power constraints (PAC) as in \cite{Vu-11}\cite{Loyka-17}, which are due to limited-power per-antenna amplifiers, interference power constraints (IPC) as in \cite{Loyka-20}, which limit the interference power induced to other users, and energy harvesting constraints (EHC) as in \cite{Hoang-20}.

The above results imply the information-theoretic optimality of popular OFDM-type transmission systems \cite{Bolcskei-02}\cite{Hoang-20} for the considered channels since such systems essentially "implement" the frequency-domain channel capacity expressions by replacing integrals with Riemann sums over subcarriers.

\textit{Notations}: lower case ($\bx$) and capital ($\bR$) bold symbols denote vectors and matrices respectively; $\bR^+$ is Hermitian conjugation; $|\bR|$, $|\bR|_F$ and $tr\bR$ are determinant, Frobenius norm  and trace of $\bR$; $(\bR)_+$ retains positive eigenmodes of Hemitian matrix $\bR$:
\ba
(\bR)_+ = \sum_{i:\gl_i(\bR)>0} \gl_i(\bR)\bu_i\bu_i^+
\ea
where $\bu_i$ and $\gl_i(\bR)$ are $i$-th eigenvector and eigenvalue of $\bR$; $\bR\ge 0$ means that $\bR$ is positive semi-definite, $(x)_+=\max(x,0)$.

\section{Channel Model and Its Capacities}
We adopt the channel model and notations from \cite{Brandenburg-74} with slight modifications, which are also consistent with \cite{Tsybakov-06}. The discrete-time channel model is as follows:
\ba
\label{eq.ch}
\by(t)=\sum_{\tau}\bH(t-\tau)\bx(\tau)+\bxi(t)
\ea
where $\by(t),\ \bx(t)$ are the output (received) and input (transmitted) $n$-dimensional vector signals, $\bxi(t)$ is correlated zero-mean Gaussian noise with memory, and $\bH(t)$ is the channel (discrete-time) impulse response $n\times n$ matrix (collecting channel impulse responses from each input to each output); $t,\ \tau = 0, \pm 1, \pm 2,...$ are discrete time variables. In this model, the noise is correlated in time as well as across outputs (receive antennas), and the channel impulse response $\bH(t)$ also introduces memory (due to e.g., multipath propagation and the related delay spread). The noise is assumed to be wide-sense stationary and hence can be characterized by its covariance matrix
\ba
\bR_{\xi}(\tau) = E\{\bxi(t)\bxi^+(t-\tau)\}
\ea
where $E\{\cdot\}$ is statistical expectation, while the channel can be represented by its (discrete) transfer function
\ba
\bH(\theta)= \sum_t \bH(t)e^{-j t \theta}
\ea
where $-\pi \le \theta \le \pi$ is the normalized frequency. Likewise, the (stationary) noise can be characterized by its (discrete) power spectral density matrix:
\ba
\bR_{\xi}(\theta)= \sum_t \bR_{\xi}(t)e^{-j t \theta}
\ea

The following assumptions have been made in \cite{Brandenburg-74}:

1. The channel is causal: $\bH(t)=0,\ t<0$, and satisfies the following conditions:
\ba
\label{eq.ch.c1}
\sum_{t=0}^{\infty} |\bH(t)|_F < \infty,\ |\bH(t)|_F < b/t,\ \forall\ t>0
\ea
for some $b< \infty$. In addition, $|\bH(\theta)| \neq 0$, $-\pi \le \theta \le \pi$, i.e., the channel is non-singular at any frequency (this condition can be relaxed later, since "singular" frequencies do not contribute to the capacity).

2. The noise covariance matrix satisfies the following:
\ba
\label{eq.ch.c2}
\sum_{t=-\infty}^{\infty} |\bR_{\xi}(t)|_F < \infty
\ea
and its PSD is also non-singular, $|\bR_{\xi}(\theta)| \neq 0$.

Note that the above conditions are not too restrictive: they essentially require the channel memory to be not too strong. Any channel with finite impulse response and with finite-memory non-singular noise satisfies them, as in state-of-the art industrial 5G channel models \cite{3GPP-22}. Under these conditions, the operational channel capacity was established in \cite{Brandenburg-74} under the total power constraint as follows:
\ba
\label{eq.C}
C = \frac{1}{4\pi} \sum_i \int_{-\pi}^{\pi} \left(\log \frac{\mu}{\gl_i(\theta)}\right)_+ d\theta
\ea
where $\gl_i(\theta)$ is $i$-th eigenvalue of $\bH^{-1}(\theta)\bR_{\xi}(\theta)(\bH(\theta)^+)^{-1}$ and $\mu>0$ is "water level" determined as a unique solution of the following equation
\ba
\label{eq.PC}
\frac{1}{2\pi} \sum_i \int_{-\pi}^{\pi} \left(\mu - \gl_i(\theta)\right)_+ d\theta = P
\ea
where $P$ is the total input (transmit) power.

On the other hand, the information capacity of this channel was established in Theorem 1 of \cite{Tsybakov-06} for $n=2$ and $\bH(\theta)=\bI$, where $\bI$ is identity matrix, including explicit expressions for the optimal input PSD matrix $\bR_{x}(\theta)$, and some special cases were considered. In the case of general $n$, a lower bound to the information capacity was obtained, and the results were further extended to $\bH(\theta)\neq \bI$.

Since the optimal input is Gaussian, the starting point of the analysis in \cite{Tsybakov-06} is the Pinsker's formula for mutual information rate under Gaussian input \cite{Pinsker-64}:
\ba
I(\bx,\by) = \frac{1}{4\pi} \int_{-\pi}^{\pi} \log \frac{|\bR_y(\theta)|}{|\bR_{\xi}(\theta)|}d\theta
\ea
which is further optimized over all input covariance matrices $\bR_{x}(\theta)$, subject to the total (average) power constraint, to obtain the information capacity $C_{inf}$:
\ba
\label{eq.Cinf}
C_{inf} = \frac{1}{4\pi} \max_{\bR_x(\theta) \in \mathcal{S}} \int_{-\pi}^{\pi} \log \frac{|\bR_y(\theta)|}{|\bR_{\xi}(\theta)|}d\theta
\ea
where $\bR_{y}(\theta)=\bR_{x}(\theta)+\bR_{\xi}(\theta)$, and the constraint set $\mathcal{S}$ represents the TPC,
\ba
\label{eq.S}
\mathcal{S} = \left\{\bR_x(\theta)\ge 0: \frac{1}{2\pi}\int_{-\pi}^{\pi} tr\bR_x(\theta) d\theta \le P \right\}
\ea
It was further conjectured, based on numerical evidence (see p. 192 in \cite{Tsybakov-06}) that the information and operational capacities are the same in all cases, $C=C_{inf}$.

We prove this conjecture in Theorem 1 below  and also obtain an explicit solution to the optimization problem in \eqref{eq.Cinf} for any $n$, which gives the optimal input PSD matrix for this channel and shows that the optimal power allocation is via water-filling (both across inputs (antennas) and frequencies).

\section{Operational Capacity = Information Capacity}

The capacity of the MIMO channel with memory in \eqref{eq.ch} under the TPC can be characterized as follows.

\begin{thm}
The operational capacity of the channel in \eqref{eq.ch} under the conditions in \eqref{eq.ch.c1} and \eqref{eq.ch.c2} is the same as its information capacity, $C=C_{inf}$, under the TPC in \eqref{eq.S}. The optimal input PSD is as follows:
\ba
\label{eq.Rx*}
\bR_x^*(\theta) = (\mu\bI - (\bH^+(\theta)\bR_{\xi}^{-1}(\theta)\bH(\theta))^{-1})_+
\ea
where  $\mu>0$  is found from the power constraint in \eqref{eq.PC}.
\end{thm}
\begin{proof}
First, observe that
\ba
\log \frac{|\bR_y(\theta)|}{|\bR_{\xi}(\theta)|} &= \log |\bI+ \bR_{\xi}^{-1}(\theta)\bH(\theta)\bR_x(\theta)\bH^+(\theta)|\\
&= \log |\bI+ \bW(\theta)\bR_x(\theta)|\\
&= \log |\bI+ \bLam_w(\theta)\bar{\bR}_x(\theta)|\\
\label{eq.log-det.4}
&\le \sum_i \log(1+ \gl_{wi}(\theta)d_{i}(\theta))
\ea
where $\bW(\theta)=\bH^+(\theta)\bR_{\xi}^{-1}(\theta)\bH(\theta)$ and $\bW(\theta)=\bU_w(\theta)\bLam_w(\theta)\bU^+_w(\theta)$ is its eigenvalue decomposition, $\bU_w(\theta)$ is the unitary matrix of its eigenvectors and $\bLam_w(\theta)$ is a diagonal matrix of its eigenvalues, $\bar{\bR}_x(\theta)=\bU_w^+(\theta)\bR_x(\theta)\bU_w(\theta)$, $\gl_{wi}(\theta)$ and $d_{i}(\theta)$ are $i$-th eigenvalue and diagonal entry of $\bW(\theta)$ and $\bar{\bR}_x(\theta)$ respectively. 1st equality follows from the channel model; 2nd and 3rd equalities follows from $|\bI+\bA\bB|=|\bI+\bB\bA|$; the inequality follows from Hadamard inequality. Further notice that $\sum_i d_i(\theta)=tr\bar{\bR}_x(\theta)= tr\bR_x(\theta)$ and hence optimizing over $\bR_x(\theta)$, $\bar{\bR}_x(\theta)$ and $d_i(\theta)$ are all the same (satisfy the same power constraint), so that
\ba
\label{eq.log-det.5}
\max_{\bR_x(\theta) \in \mathcal{S}} &\int_{-\pi}^{\pi} \log \frac{|\bR_y(\theta)|}{|\bR_{\xi}(\theta)|}d\theta \notag\\
&\le \sum_i \max_{d_i(\theta)\ge 0} \int_{-\pi}^{\pi} \log(1+ \gl_{wi}(\theta)d_{i}(\theta))d\theta\\ \notag
&\qquad\qquad \mbox{s.t.}\ \frac{1}{2\pi}\int_{-\pi}^{\pi}\sum_i d_{i}(\theta)d\theta \le P
\ea
Since the upper bound in \eqref{eq.log-det.5} is the information rate of $n$ parallel Gaussian channels, its capacity is attained by the standard water-filling solution (see e.g., Theorem 8.5.1 in \cite{Gallager}; its slight extension applies to the parallel channel setting as well):
\ba
d_i(\theta) =\gl_i(\bR_x(\theta)) = (\mu - \gl_{wi}^{-1}(\theta))_+
\ea
where $\mu>0$ is found from the total power constraint:
\ba
\frac{1}{2\pi}\int_{-\pi}^{\pi}\sum_i (\mu - \gl_{wi}^{-1}(\theta))_+ d\theta = P
\ea
which maximizes the upper bound in \eqref{eq.log-det.5} under the input  power constraint. Further note that the equality in \eqref{eq.log-det.5} is attained when $\bR_x(\theta)$ and $\bW(\theta)$ have the same eigenvectors (this can always been done since the power constraint does not limit the eigenvectors of $\bR_x(\theta)$ but only its eigenvalues). Hence, the optimal input PSD matrix has the same eigenvectors as those of $\bW(\theta)$ and can be expressed as
\ba
\label{eq.Rx*.2}
\bR_x^*(\theta) &=(\mu\bI - \bW^{-1}(\theta))_+ \notag\\
&=  (\mu\bI - (\bH^+(\theta)\bR_{\xi}^{-1}(\theta)\bH(\theta))^{-1})_+
\ea
Finally, the information capacity is
\ba\notag
C_{inf} &= \frac{1}{4\pi} \int_{-\pi}^{\pi} \log |\bI+ \bH^+(\theta)\bR_{\xi}^{-1}(\theta)\bH(\theta)\bR_x^*(\theta)| d\theta\\ \notag
&= \frac{1}{4\pi} \sum_i \int_{-\pi}^{\pi} \log(1+\gl_{wi}(\theta)(\mu - \gl_{wi}^{-1}(\theta))_+) d\theta\\
\label{eq.Cinf.3}
&= \frac{1}{4\pi} \sum_i \int_{\theta: \mu \gl_{wi}>1} \log(\mu\gl_{wi}(\theta)) d\theta
\ea
which is exactly the same as in \eqref{eq.C}, since $\gl_i(\theta)=\gl_{wi}^{-1}(\theta)$.
\end{proof}

We remark that this result also proves (indirectly) that the above channel is information stable (under the stated assumptions), since the operational and information capacities coincide only if the channel is information-stable \cite{Dobrushin-59}\cite{Ting}.

This result can also be used when $\bH(\theta)$ is singular, i.e., $|\bH(\theta)|=0$ for some $\theta$, which corresponds to $\bW(\theta)$ being singular: since $(\cdot)_+$ operator eliminates negative eigenmodes, zero eigenvalues of $\bW(\theta)$ do not affect $\bR_x^*(\theta)$, which assigns zero input power at those frequencies (this can be seen via the standard continuity argument), as it should be. This is also clear from \eqref{eq.Cinf.3}, where the integration is over the region where $\gl_{wi}(\theta)>0$, which corresponds to $\bH(\theta)$ being non-singular at those frequencies.

With minor modifications, these results can also be extended to the case of unequal number of inputs and outputs (antennas). In particular, \eqref{eq.Rx*.2}-\eqref{eq.Cinf.3} apply directly to this case.

In the special case of $\bH(\theta)=\bI$ (i.e., parallel channels with memoryless impulse response, but where the noise can be correlated and with memory), \eqref{eq.Rx*} reduces to
\ba
\label{eq.Rx*.I}
\bR_x^*(\theta) = (\mu\bI - \bR_{\xi}(\theta))_+
\ea
which coincides with Theorem 1 in \cite{Tsybakov-06} for $n=2$ and further extends it to any $n>2$.

\section{Joint Power Constraints}

While the constraint set in \eqref{eq.S} includes only the TPC, Theorem 1 can be extended to include additional additional power constraints as well. Among the most important ones are per-antenna constraints as in \cite{Vu-11}\cite{Loyka-17}, interference power constraints (typical for multi-user systems including cognitive radio) as in \cite{Loyka-20} as well as energy harvesting  constraints. While the above constraints were formulated for frequency-flat AWGN channels, we give below their extension to frequency-selective channels (or, equivalently, channels with memory).

Per-antenna power constraints (PAC) limit the average power radiated by each antenna (due to e.g., limited power amplifiers):
\ba
\label{eq.PAC}
\frac{1}{2\pi}\int_{-\pi}^{\pi} r_{ii}(\theta) d\theta \le P_i
\ea
where $r_{ii}(\theta) \ge 0$ is $i$-th diagonal entry of $\bR_x(\theta)$, the integral represents the average power radiated by $i$-th antenna and $P_i$ is its maximum value.

The interference power constraints (IPC) take the following form:
\ba
\label{eq.IPC}
\frac{1}{2\pi}\int_{-\pi}^{\pi} tr\{\bH_k(\theta)\bR_x(\theta)\bH_k(\theta)^+\} d\theta \le P_{I,k}
\ea
where $\bH_k(\theta)$ represents the channel to $k$-th user and $P_{I,k}$ is the maximum interference power induced to that user.

The energy harvesting constraint (EHC) is opposite of the IPC:
\ba
\label{eq.EHC}
\frac{1}{2\pi}\int_{-\pi}^{\pi} tr\{\bH_m(\theta)\bR_x(\theta)\bH_m(\theta)^+\} d\theta \ge P_{E,m}
\ea
where $\bH_m(\theta)$ represents the channel to $m$-th energy-harvesting user and $P_{E,m}$ is the minimum harvested power (or energy per unit time) of that user.

The overall (joint) constraint set is
\ba
\label{eq.S.joint}
\mathcal{S}_{o} = \left\{\bR_x(\theta)\ \ \mbox{s.t.}\ \ \eqref{eq.S}, \eqref{eq.PAC}, \eqref{eq.IPC}, \eqref{eq.EHC} \right\}
\ea
where some constraints can be omitted, if necessary. We will further assume that this set is not empty, i.e., the constraints are compatible (otherwise, the capacity is zero). The following Theorem is an extension of Theorem 1 to the case of the joint constraints.

\begin{thm}
Under the joint constraints in \eqref{eq.S.joint}, the operational capacity of the channel in Sec. II is the same as its information capacity, $C= C_{inf}$, where
\ba
\label{eq.Cinf.join}
C_{inf} \triangleq \frac{1}{4\pi} \max_{\bR_x(\theta) \in \mathcal{S}_{o}} \int_{-\pi}^{\pi} \log |\bI+ \bW(\theta)\bR_x(\theta)| d\theta
\ea
\end{thm}
\begin{proof}
To prove the converse, note that it follows from \cite[Theorem 3.5.2 and 3.6.1]{Han-03} that $C \le C_{inf}$. Alternatively, one can use the converse in \cite[Sec. 3.1]{Brandenburg-74} and observe that it also holds under the joint constraints since its key ingredient, Fano's inequality, is not affected by the constraints but only the maximization of mutual information is. Likewise, to establish achievability, one can use the codebook construction in \cite[eq. (3) and Sec. 3.2]{Brandenburg-74} and amend it with the PAC, IPC and EHC, in addition to the TPC; all other steps of the proof remain unaffected (since they do not depend on the constraints).
\end{proof}

While the optimal input covariance matrix $\bR_x^*(\theta)$ is in a closed-form in Theorem 1, a closed-form solution to the maximization in \eqref{eq.Cinf.join} under the joint constraints is not known in the general case, even for the memoryless channel.

However, it can be further simplified to a more explicit form in some special cases. Let us consider the PAC alone, as  in \eqref{eq.PAC}, and  extend \cite{Vu-11} to the case of a MISO channel with memory, $\bH(\theta)=\bh(\theta)^+$, where the noise is i.i.d., $\bR_{\xi}=\sigma_0^2\bI$, but the channel has memory via its impulse response $\bh(t)$ (e.g., due to multipath propagation). Using the same arguments as in \cite{Vu-11} but applied to the frequency-domain channel $\bh(\theta)$, it follows that the optimal covariance in \eqref{eq.Cinf.join} is of rank-one, $\bR_x^*(\theta)=\bw(\theta)\bw(\theta)^+$, where
\ba
w_i(\theta)=\sqrt{r_{ii}(\theta)}e^{j\phi_i(\theta)},\ \phi_i(\theta)=\arg\{h_i(\theta)\}
\ea
and the maximization in \eqref{eq.Cinf.join} reduces to the optimal power allocation in the frequency domain:
\ba
\label{eq.Cinf.PAC}
C &=  \max_{\{r_{ii}(\theta)\}}\frac{1}{4\pi} \int_{-\pi}^{\pi} \log \big(1+ \sigma_0^{-2}\big|\sum_i |h_i(\theta)|\sqrt{r_{ii}(\theta)}\big|^2\big) d\theta \notag\\
&\qquad \mbox{s.t.}\ \frac{1}{2\pi}  \int_{-\pi}^{\pi} r_{ii}(\theta) d\theta \le P_i,\ r_{ii}(\theta)\ge 0
\ea
To the best of our knowledge, no closed-form solution is known for this problem (note that the standard water-filling solution, derived under the TPC, does not apply here due to the PAC).

From an engineering perspective, Theorems 1 and 2 establish the information-theoretic optimality of OFDM-type transmission systems operating over channels with memory as in \eqref{eq.ch}, since such systems essentially "implement" the capacity expressions in \eqref{eq.Cinf}, \eqref{eq.Cinf.join}, \eqref{eq.Cinf.PAC} by replacing the integrals with Riemann sums, e.g., $C_{inf}$ in \eqref{eq.Cinf.join} is approximated by
\ba
\label{eq.Cinf.join.A}
C_{inf} \approx \frac{1}{4\pi} \max_{\bR_x(\theta_i)} \sum_{i} \log |\bI+ \bW(\theta_i)\bR_x(\theta_i)| \Delta\theta_i
\ea
where $\theta_i$ represents $i$-th subcarrier,  $\Delta\theta_i$ represents its bandwidth and the sum is over all subcarriers. While such systems are widely studied in the literature, see e.g., \cite{Bolcskei-02}\cite{Hoang-20}, their information rates are often evaluated via respective mutual information and it remains unclear whether (i) those rates are indeed achievable (this is not trivial since channels with memory are not necessarily information-stable and hence their mutual information may have no operational meaning \cite{Dobrushin-59}-\cite{Ting}) and whether (ii) they can be further improved.  The above Theorems provide the affirmative answer to (i) and the negative answer to (ii), both for the channel in \eqref{eq.ch}.

\section{Concusion}

A Gaussian MIMO channel with memory has been considered and its operational channel capacity hes been obtained in a closed form under the total power constraint, and was shown to be equal to its information capacity, thus proving the earlier conjecture in \cite{Tsybakov-06}. This extends the seminal result in \cite{Tsybakov-65}\cite{Telatar-95} to Gaussian MIMO channels  with memory and implies an information-theoretic optimality of OFDM-type transmission systems for such channels. This result is further extended to the case of joint power constraints, including per-antenna constraints (either alone or in combination with the total power constraints) as well as interference and energy-harvesting constraints. For memoryless channels, the reported results reduce to the well-known capacity expressions, as it should be.


\end{document}